\documentclass[11pt]{article}

\usepackage[a4paper,margin=1in]{geometry}
\usepackage{amsmath,amssymb,amsthm,mathtools}
\usepackage{mathrsfs}
\usepackage{enumitem}
\usepackage{hyperref}
\usepackage{microtype}

\hypersetup{
  colorlinks=true,
  linkcolor=blue,
  citecolor=blue,
  urlcolor=blue
}

\newtheorem{theorem}{Theorem}[section]
\newtheorem{proposition}[theorem]{Proposition}
\newtheorem{lemma}[theorem]{Lemma}
\newtheorem{corollary}[theorem]{Corollary}
\newtheorem{remark}[theorem]{Remark}
\theoremstyle{definition}

\newcommand{\supp}{\operatorname{supp}}
\newcommand{\norm}[1]{\left\lVert #1\right\rVert}

\newcommand{\im}{\operatorname{Im}}

\newcommand{\e}{\mathrm e}

\title{Inverse Resonance Problem for the Half-Line Schr\"odinger Operator\\
with a Complex Robin Boundary Condition}
\author{Roman Shterenberg}
\date{}

\begin{document}

\maketitle

\begin{abstract}
We study the inverse resonance problem for the one-dimensional Schr\"odinger
operator on the half-line with a compactly supported complex-valued potential
and a complex Robin boundary condition
\[
u'(0)-a u(0)=0,\qquad a\in\mathbb C.
\]
The proof follows the transformation-operator, entire-function, and Fourier
analytic strategy developed by Marletta, Shterenberg and Weikard for the
Dirichlet problem \cite{MSW}.  The Dirichlet Jost function
$\psi(z,0)$ is replaced by the Robin characteristic function
\[
F_a(z;q)=\psi_x(z,0;q)-a\psi(z,0;q).
\]
We establish the corresponding large-energy estimates, Hadamard
factorization, finite-data stability, and uniqueness results.  The principal
new feature is that $F_a(z;q)=iz-a+O(1)$ rather than $1+O(z^{-1})$.
In the transformation-operator argument the relevant boundary quantity is
$B_x(0,t)-aB(0,t)$, together with the scalar term $B(0,0)$.  For fixed
$a\in\mathbb C$, and uniformly for $a$ in bounded subsets of $\mathbb C$,
the resulting inverse resonance estimates have the same qualitative
form as those proved in \cite{MSW} for the Dirichlet problem.  In particular, for a fixed Robin parameter $a$, the complete set of Robin
eigenvalues and resonances, counted with multiplicity, determines the potential
uniquely.

The text is produced with the help of AI ChatGPT for Academy.
\end{abstract}

\section{Introduction}

Consider
\begin{equation}
  -u''(x)+q(x)u(x)=z^2u(x),\qquad x>0,
  \label{eq:schrodinger}
\end{equation}
where $q$ is a complex-valued potential supported in $[0,1]$.  We impose the
Robin boundary condition
\begin{equation}
  u'(0)-a u(0)=0,\qquad a\in\mathbb C.
  \label{eq:robin}
\end{equation}

The inverse resonance problem for compactly supported half-line
Schr\"odinger operators has been studied extensively.  The present paper is
an extension of the approach of Marletta, Shterenberg and Weikard
\cite{MSW}, which treats the half-line problem with the \emph{Dirichlet}
boundary condition
\[
u(0)=0.
\]
We shall refer to \cite{MSW} throughout as the ``Dirichlet paper''.
The transformation-operator construction, Jensen estimates, canonical
products, Fourier estimates, and Volterra analysis used below are based on
that paper.

For $Q>0$, let
\begin{equation}
 B(Q)=
 \left\{
 q\in L^1(0,\infty):
 \supp q\subset[0,1],\quad
 \norm{q}_{L^1(0,1)}\le Q
 \right\}.
 \label{eq:BQ}
\end{equation}
The potentials in $B(Q)$ may be complex-valued, as in \cite{MSW}.

Let $\psi(z,x)$ denote the Jost solution, characterized by
\begin{equation}
 \psi(z,x)=\e^{izx},\qquad x\ge1.
 \label{eq:jost}
\end{equation}
For fixed $x\ge0$, $\psi(z,x)$ is entire of order at most one in $z$.
The transformation representation used in \cite[Section 2]{MSW} is
\begin{equation}
 \psi(z,x)=
 \e^{izx}
 +\int_x^{2-x}K_q(x,t)\e^{izt}\,dt.
 \label{eq:transformation}
\end{equation}

For the Dirichlet problem the characteristic function is $\psi(z,0)$.
In the present paper we replace it by
\begin{equation}
 \boxed{
 F_a(z;q):=\psi_x(z,0;q)-a\psi(z,0;q).
 }
 \label{eq:Fa}
\end{equation}
The purpose of the paper is to establish the analogue of the inverse
resonance results of \cite{MSW} for the characteristic function
\eqref{eq:Fa}, allowing $a$ to be complex.

There are two differences between the two problems.  First,
\begin{equation}
 F_a(z;q)=iz-a+O(1),
 \qquad |z|\to\infty,
 \label{eq:leading}
\end{equation}
in the high-energy region, whereas the Dirichlet Jost function is
normalized to $1$ at high energy.  Secondly, when the transformation
relation is differentiated at $x=0$, the relevant boundary expression is
\[
 B_x(0,t)-aB(0,t),
\]
rather than simply $B(0,t)$.  Once these modifications are made, the
analytic part of the proof follows the same structure as in
\cite{MSW}.

\section{Jost solutions and the Robin characteristic function}

Let $q\in B(Q)$.  The Jost solution is the unique solution of
\eqref{eq:schrodinger} satisfying \eqref{eq:jost}.  The transformation
operator gives \eqref{eq:transformation}.  In particular,
\begin{equation}
 \psi(z,0)
 =
 1+\int_0^2K_q(0,t)\e^{izt}\,dt.
 \label{eq:psi0}
\end{equation}

The characteristic function for the Robin problem is $F_a$ in
\eqref{eq:Fa}.  Since $\psi(z,0)$ and $\psi_x(z,0)$ are entire functions
of $z$, $F_a$ is entire.

\begin{proposition}
Let $q\in B(Q)$ and $a\in\mathbb C$.  If $\im z>0$, then
\[
 F_a(z;q)=0
\]
if and only if $z^2$ is an eigenvalue of the half-line realization of
\eqref{eq:schrodinger} subject to \eqref{eq:robin}.  Zeros in the lower
half-plane are resonances in the usual meromorphic continuation sense.
\end{proposition}

\begin{proof}
If $\im z>0$, then the Jost solution decays exponentially at infinity and
therefore belongs to $L^2(0,\infty)$.  The equation
$F_a(z;q)=0$ is exactly
\[
\psi_x(z,0;q)-a\psi(z,0;q)=0,
\]
which is the Robin boundary condition.  Hence $\psi(z,\cdot)$ is an
eigenfunction.

Conversely, an $L^2$ solution for $\im z>0$ is a scalar multiple of the
Jost solution, so the Robin boundary condition implies $F_a(z;q)=0$.
The resonance statement follows from the standard analytic continuation
of the Jost solution, exactly as in the Dirichlet discussion in
\cite[Section 1]{MSW}.
\end{proof}

\section{High-energy estimates}

Differentiating the transformation representation, whenever the boundary
trace is represented by the corresponding derivative of the transformation
kernel, gives
\begin{equation}
 \psi_x(z,0)
 =
 iz-K_q(0,0)
 +\int_0^2K_{q,x}(0,t)\e^{izt}\,dt.
 \label{eq:psix}
\end{equation}
Equivalently, the estimates below may be obtained directly from the
differentiated Volterra equation, which avoids imposing unnecessary
pointwise regularity on $K_{q,x}$.

Combining \eqref{eq:psi0} and \eqref{eq:psix},
\begin{equation}
 F_a(z;q)
 =
 iz-a-K_q(0,0)
 +\int_0^2
 \left(K_{q,x}(0,t)-aK_q(0,t)\right)\e^{izt}\,dt.
 \label{eq:Frepresentation}
\end{equation}

Thus the free characteristic function is
\begin{equation}
 F_a^{(0)}(z)=iz-a.
 \label{eq:free}
\end{equation}

\begin{lemma}
Let $Q,A>0$.  There exists $C=C(Q,A)$ such that, for
$q\in B(Q)$ and $|a|\le A$,
\begin{align}
 |F_a(x;q)|&\le C(1+|x|),\qquad x\in\mathbb R,
 \label{eq:realbound}\\
 |F_a(z;q)|&\le C(1+|z|)\e^{2|z|},\qquad z\in\mathbb C.
 \label{eq:complexbound}
\end{align}
Moreover, for $\rho\ge\rho_0(Q,A)$ and
\[
 |z-3i\rho|\le\rho,
\]
one has
\begin{equation}
 \left|
 \frac{F_a(z;q)}{iz-a}-1
 \right|
 \le \frac{C(Q,A)}{\rho}.
 \label{eq:normalizedbound}
\end{equation}
\end{lemma}

\begin{proof}
The estimates follow from the transformation representation and the
corresponding estimates for the Jost solution used in
\cite[Lemma 3.2]{MSW}.  The only additional contribution is the explicit
linear term $iz-a$.  On the circle centered at $3i\rho$ the integral
terms in \eqref{eq:Frepresentation} are uniformly $O(Q)$ (the interval
of integration has fixed length), while $|iz-a|\ge c\rho$ for $\rho$
sufficiently large relative to $A$.  After division by $iz-a$ these
integral terms are therefore $O(\rho^{-1})$, which gives
\eqref{eq:normalizedbound}.
\end{proof}

\section{Hadamard factorization}

Since $F_a$ is entire of order at most one, Hadamard's theorem gives
\begin{equation}
 F_a(z;q)
 =
 z^{m_a}\e^{g_a(z)}
 \prod_n E\left(\frac{z}{z_n}\right),
 \qquad
 E(w)=(1-w)\e^w,
 \label{eq:hadamard}
\end{equation}
where $z_n$ are the nonzero zeros of $F_a$, counted with multiplicity,
and
\[
 g_a(z)=c_1z+c_0.
\]

The expression $iz-a$ is used only as an asymptotic normalization; it is
not asserted to be a factor of $F_a$.  In particular, the point $z=-ia$ is
the zero of the free characteristic function, but it need not be a zero of
$F_a(z;q)$.  Thus $F_a(z;q)/(iz-a)$ is meromorphic, with a possible simple
pole at $z=-ia$, and is used below only in the high-energy region where
$|iz-a|\asymp |z|$.  In that region,
\[
 \frac{F_a(z;q)}{iz-a}=1+O(z^{-1}).
\]

\begin{lemma}
For every $Q,A>0$ there exist $C=C(Q,A)$ and $R_0=R_0(Q,A)$ such that
the number $N_a(r)$ of zeros of $F_a$ in $|z|<r$ satisfies
\begin{equation}
 N_a(r)\le C(1+r),\qquad r\ge1.
 \label{eq:zero-count}
\end{equation}
Moreover, if
\[
 \Pi_a(R,z)
 =
 \prod_{|z_n|\ge R}
 E\left(\frac z{z_n}\right),
\]
then, for $|z|\le R/2$,
\begin{equation}
 |\Pi_a(R,z)-1|
 \le
 \frac{C(Q,A)|z|^2}{R}
 \exp\left(\frac{C(Q,A)|z|^2}{R}\right).
 \label{eq:tailproduct}
\end{equation}
\end{lemma}

\begin{proof}
The proof is the Jensen-formula argument of \cite[Lemma 3.2 and
the subsequent canonical-product estimates]{MSW}.  The factor $iz-a$
contributes only logarithmically to the maximum-modulus estimates.
For $|a|\le A$, this changes the constants but not the order or the
linear zero-counting bound.  The estimate \eqref{eq:tailproduct} then
follows from the standard estimate
$|E(w)-1|\le C|w|^2$ for $|w|\le1/2$.
\end{proof}

\section{Finite-data stability for the Robin characteristic function}

We now prove the direct analogue of the finite-data estimate in
\cite[Theorem 3.1]{MSW}.

\begin{theorem}
Let $Q,A>0$.  There exist constants $C=C(Q,A)$ and
$R_0=R_0(Q,A)$ with the following property.

Let $q,\widetilde q\in B(Q)$ and $a\in\mathbb C$, $|a|\le A$.  Suppose
that the zeros of
\[
 F_a(z;q)
 \quad\text{and}\quad
 F_a(z;\widetilde q)
\]
coincide in $|z|<R$, including multiplicities.  Then, for $R\ge R_0$,
\begin{equation}
 \boxed{
 \sup_{|x|\le R^{1/3}}
 \frac{|F_a(x;q)-F_a(x;\widetilde q)|}{1+|x|}
 \le C R^{-1/3}.
 }
 \label{eq:finite-stability}
\end{equation}
\end{theorem}

\begin{proof}
Write the Hadamard factorizations in the form
\[
 F_a(z;q)=z^{m}e^{g(z)}P_{<R}(z)P_{\ge R}(z),
 \qquad
 F_a(z;\widetilde q)=z^{m}e^{\widetilde g(z)}
 \widetilde P_{<R}(z)\widetilde P_{\ge R}(z),
\]
where the common zeros in $|z|<R$, including their multiplicities, have
been collected in $P_{<R}=\widetilde P_{<R}$ and
$g-\widetilde g$ is affine.  The integer $m$ is the common multiplicity of
any zero at the origin.  Notice that the free factor $iz-a$ is not inserted
into either Hadamard product; it is used only for the high-energy
normalization in \eqref{eq:normalizedbound}.

Consequently,
\[
 \frac{F_a(z;q)}{F_a(z;\widetilde q)}
 =e^{g(z)-\widetilde g(z)}
   \frac{P_{\ge R}(z)}{\widetilde P_{\ge R}(z)}
\]
whenever the denominator is nonzero.  On
$|z-3i\rho|\le\rho$, the normalized estimate
\eqref{eq:normalizedbound} gives
\[
 \frac{F_a(z;q)}{F_a(z;\widetilde q)}=1+O(\rho^{-1}),
\]
while the tail-product estimate gives
\[
 P_{\ge R}(z)=1+O(\rho^2/R),\qquad
 \widetilde P_{\ge R}(z)=1+O(\rho^2/R),
\]
provided $\rho\le R/4$.  For $R$ large these estimates imply
$|\widetilde P_{\ge R}(z)|\ge1/2$, so division is legitimate and the
same bound holds for their quotient.  Hence
\[
 e^{g(z)-\widetilde g(z)}-1
 =O\left(\rho^{-1}+\frac{\rho^2}{R}\right)
\]
on that disc.  Put $h=g-\widetilde g=c_1z+c_0$.  On the high-imaginary disc the
quotient is nonzero and is $1+O(\rho^{-1})$, so a single analytic branch of
$\log$ may be chosen there.  The normalization at infinity fixes the branch,
and hence $h(z)=O(\rho^{-1})$ on that disc.  Since $h$ is affine, evaluating
at two points of the disc gives $|c_1|\rho+|c_0|\le C\rho^{-1}$, and
therefore $h(z)=O(\rho^{-1})$ throughout $|z|\le\rho$.

Choose $\rho=R^{1/3}$.  Then
\[
 \rho^{-1}+\frac{\rho^2}{R}\le C R^{-1/3}.
\]
The preceding estimates therefore give
\[
 \frac{F_a(x;q)}{F_a(x;\widetilde q)}=1+O(R^{-1/3})
\]
for real $|x|\le R^{1/3}$, away from the finitely many common zeros; the
identity extends across those zeros by analyticity.  Hence, using
\eqref{eq:realbound},
\[
 |F_a(x;q)-F_a(x;\widetilde q)|
 \le C(1+|x|)R^{-1/3},
 \qquad |x|\le R^{1/3},
\]
which is \eqref{eq:finite-stability}.  This argument uses $iz-a$ only through the
normalization estimate and never treats $iz-a$ as a zero factor of $F_a$.
\end{proof}

\begin{remark}
The factor $(1+|x|)$ in \eqref{eq:finite-stability} is natural and should
not be omitted: the Robin characteristic function has linear high-energy
growth, whereas the Dirichlet Jost function is asymptotic to $1$.
\end{remark}

\section{Transformation kernels and the Robin boundary trace}

Let $\psi$ and $\widetilde\psi$ correspond to $q$ and $\widetilde q$.
As in \cite[Section 2]{MSW}, there is a transformation kernel $B$ such
that
\begin{equation}
 \widetilde\psi(z,x)
 =
 \psi(z,x)
 +\int_x^{2-x}B(x,t)\psi(z,t)\,dt.
 \label{eq:Btrans}
\end{equation}
It satisfies the diagonal identity
\begin{equation}
 \boxed{
 2B(x,x)
 =
 \int_x^1(\widetilde q(s)-q(s))\,ds.
 }
 \label{eq:diagonal}
\end{equation}

Differentiating \eqref{eq:Btrans} at $x=0$ and subtracting $a$ times
the value at $x=0$ gives
\begin{equation}
 \boxed{
 F_a(z;\widetilde q)-F_a(z;q)
 =
 -B(0,0)\psi(z,0)
 +\int_0^2 C_a(t)\psi(z,t)\,dt,
 }
 \label{eq:RobinB}
\end{equation}
where
\begin{equation}
 \boxed{
 C_a(t):=B_x(0,t)-aB(0,t).
 }
 \label{eq:Ca}
\end{equation}

Equation \eqref{eq:RobinB} is the precise modification of the boundary
identity used in the Dirichlet paper.  In particular, it is not sufficient
simply to replace $\psi(z,0)$ by $F_a(z)$ in the Fourier argument.

The kernel $B(0,t)$ is related to the original transformation kernels by
\begin{equation}
 B(0,t)
 =
 K_{\widetilde q}(0,t)-K_q(0,t)
 +
 \int_0^t
 \bigl(K_{\widetilde q}(0,s)-K_q(0,s)\bigr)L_q(s,t)\,ds,
 \label{eq:B0}
\end{equation}
as in \cite{MSW}.

% ------------------------------------------------------------------
% Corrected stability chain.
% The factor (1+|z|) in the finite-frequency estimate is retained.
% The cutoff is R^\alpha with alpha=p/[3(3p-1)].
% ------------------------------------------------------------------

\section{Fourier estimate}

The linear high-energy growth of the Robin characteristic function must
be retained in the Fourier argument.  Put
\[
D_a(z):=F_a(z;q)-F_a(z;\widetilde q).
\]
By \eqref{eq:finite-stability},
\begin{equation}
 |D_a(z)|
 \le C(1+|z|)R^{-1/3},
 \qquad |z|\le R^{1/3}.
 \label{eq:Robin-low-frequency}
\end{equation}

For \(1<p\le2\), define
\begin{equation}
 \alpha=\frac{p}{3(3p-1)},
 \qquad
 \beta=\frac{\alpha(p-1)}{p}
       =\frac{p-1}{3(3p-1)}.
 \label{eq:alpha-beta}
\end{equation}
The choice of \(\alpha\) is dictated by the factor \(1+|z|\) in
\eqref{eq:Robin-low-frequency}.  Indeed,
\[
-\frac13+2\alpha=-\beta.
\]

\begin{lemma}[Robin Fourier decomposition]
\label{lem:Robin-Fourier-decomp}
There is a function $\eta_a\in L^p(0,2)$, depending on $q,\widetilde q$
(and on $a$), such that, for real $z$,
\begin{equation}
 D_a(z)=B(0,0)+\int_0^2 \eta_a(t)e^{izt}\,dt,
 \label{eq:Robin-Fourier-decomp}
\end{equation}
where $D_a(z)=F_a(z;q)-F_a(z;\widetilde q)$ and
\begin{equation}
 \|\eta_a\|_{L^p(0,2)}
 \le C(Q,A,p)\bigl(1+\|\widetilde q-q\|_{L^p(0,1)}\bigr).
 \label{eq:Robin-ha-bound}
\end{equation}
Moreover,
\begin{equation}
 2B(0,0)=\int_0^1(\widetilde q(s)-q(s))\,ds.
 \label{eq:B00-integral}
\end{equation}
\end{lemma}

\begin{proof}
We use the transformation representation with $q$ as the reference
potential.  Thus $\psi$ denotes the Jost solution for $q$, while
$\widetilde\psi$ is obtained from $\psi$ by the transformation kernel
$B$:
\begin{equation}
 \widetilde\psi(z,x)=\psi(z,x)
 +\int_x^{2-x}B(x,t)\psi(z,t)\,dt.
 \label{eq:Robin-Fourier-proof-B}
\end{equation}
At $x=0$, differentiation of this identity gives the boundary relation
\eqref{eq:RobinB}.  The important point is that the first term in
\eqref{eq:RobinB} cannot be discarded: it is the scalar
$-B(0,0)\psi(z,0)$.

For the Fourier decomposition it is preferable to return to the
transformation representation for the individual Jost solutions.  At
$x=0$ we have
\[
 \psi(z,0)=1+\int_0^2K_q(0,t)e^{izt}\,dt,
\]
and, in the weak sense sufficient for the Fourier inversion argument,
\[
 \psi_x(z,0)=iz-K_q(0,0)
 +\int_0^2K_{q,x}(0,t)e^{izt}\,dt.
\]
The corresponding identities hold for $\widetilde q$.  Hence
\begin{align}
 D_a(z)
 &=B(0,0)
   +\int_0^2 \eta_a(t)e^{izt}\,dt,
 \label{eq:Robin-D-decomp}
\end{align}
where, with the signs determined by the definition
$D_a=F_a(q)-F_a(\widetilde q)$,
\begin{equation}
 \eta_a(t)=
 K_{q,x}(0,t)-K_{\widetilde q,x}(0,t)
 -a\bigl(K_q(0,t)-K_{\widetilde q}(0,t)\bigr).
 \label{eq:Robin-ha-explicit}
\end{equation}
Here we have used
$K_{\widetilde q}(0,0)-K_q(0,0)=B(0,0)$, which follows from
\eqref{eq:B0} by putting $t=0$.

We prove the required trace estimate directly from the characteristic-coordinate
Volterra representation of the transformation kernel.  This avoids any
appeal to a nonexistent pointwise derivative of $q$.  Put

\[
 K_q(x,t)=K_{0,q}(x,t)+H_q(x,t),
 \qquad
 K_{0,q}(x,t)=\frac12\int_{(x+t)/2}^{1}q(s)\,ds .
\]
In the characteristic coordinates
\[
 \xi=\frac{x+t}{2},\qquad \eta=\frac{t-x}{2},
\]
the remainder satisfies
\[
 (H_q)_{\xi\eta}(\xi,\eta)
 =-q(\xi-\eta)K_q(\xi-\eta,\xi+\eta),
 \qquad H_q(\xi,0)=H_q(1,\eta)=0.
\]
Consequently,
\begin{equation}
 H_q(\xi,\eta)
 =\int_\xi^1\int_0^\eta
 q(s-v)K_q(s-v,s+v)\,dv\,ds .
 \label{eq:H-characteristic-representation}
\end{equation}
The usual Volterra iteration gives
\begin{equation}
 \sup_{0\le x\le1,\ x\le t\le2-x}|K_q(x,t)|
 \le C(Q).
 \label{eq:K-uniform-bound}
\end{equation}
Differentiating \eqref{eq:H-characteristic-representation} with respect to
its integration limits is legitimate for almost every $t$: the resulting
boundary traces are absolutely continuous in $t$, and all integrands are
locally integrable because $q\in L^1$ and $K_q$ is bounded.  Since
$x=\xi-\eta$ and $t=\xi+\eta$, we have
$\partial_x=\frac12(\partial_\xi-\partial_\eta)$.  Thus, on the boundary
$x=0$ (so $\xi=\eta=t/2$),
\begin{align}
 H_{q,x}(0,t)
 ={}&-\frac12\int_0^{t/2}
 q(t/2-v)K_q(t/2-v,t/2+v)\,dv \\
 &-\frac12\int_{t/2}^{1}
 q(s-t/2)K_q(s-t/2,s+t/2)\,ds .
 \label{eq:Hx-explicit}
\end{align}
In particular, the first term in $K_q$ is a genuine primitive of $q$:
\[
 K_{0,q}(x,t)=\frac12\int_{(x+t)/2}^{1}q(s)\,ds,
 \qquad
 K_{0,q,x}(0,t)=-\frac14q(t/2)\quad\text{a.e.}
\]
Thus there is no differentiability assumption on $q$ beyond $q\in L^1$.
Using \eqref{eq:K-uniform-bound} and Minkowski's integral inequality (or,
equivalently, the boundedness of the corresponding Volterra operators),
we obtain
\begin{equation}
 \|H_{q,x}(0,\cdot)\|_{L^p(0,2)}
 \le C(Q,p)\|q\|_{L^p(0,1)}
 \le C(Q,p)Q,
 \label{eq:Hx-Lp}
\end{equation}
for $1\le p\le\infty$.  More importantly, subtracting the corresponding
formulae for $q$ and $\widetilde q$ gives
\begin{equation}
 \|H_{\widetilde q,x}(0,\cdot)-H_{q,x}(0,\cdot)\|_{L^p(0,2)}
 \le C(Q,p)\|\widetilde q-q\|_{L^p(0,1)},
 \label{eq:Hx-difference-Lp}
\end{equation}
where we used, in addition to \eqref{eq:K-uniform-bound}, the standard
Volterra estimate
\[
 \|K_{\widetilde q}-K_q\|_{L^\infty(\triangle)}
 \le C(Q)\|\widetilde q-q\|_{L^1(0,1)}.
\]
Indeed, after subtraction each integrand is a sum of a term containing
$\widetilde q-q$ times one bounded transformation kernel and a term
containing $q$ times $K_{\widetilde q}-K_q$; the displayed estimate and
$\|\widetilde q-q\|_1\le\|\widetilde q-q\|_p$ give the stated bound.
Finally,
\[
 K_{q,x}(0,t)= -\frac14q(t/2)+H_{q,x}(0,t),
\]
so \eqref{eq:Hx-Lp} and \eqref{eq:Hx-difference-Lp} imply
\begin{align}
 \|K_{q,x}(0,\cdot)\|_{L^p(0,2)}&\le C(Q,p),\\
 \|K_{\widetilde q,x}(0,\cdot)-K_{q,x}(0,\cdot)\|_{L^p(0,2)}
 &\le C(Q,p)\|\widetilde q-q\|_{L^p(0,1)}.
 \label{eq:Robin-Kx-trace-estimate}
\end{align}
The $L^p$ estimate for the undifferentiated trace follows in the same way
from the representation of $K_q$ and is bounded by $C(Q)$.  Therefore
\eqref{eq:Robin-ha-explicit} satisfies
\eqref{eq:Robin-ha-bound}, with the slightly stronger dependence
$C(Q,A,p)(1+\|\widetilde q-q\|_p)$.

Finally, the diagonal identity for the transformation kernel gives
\[
 2B(x,x)=\int_x^1(\widetilde q(s)-q(s))\,ds,
\]
and $x=0$ gives \eqref{eq:B00-integral}.
\end{proof}

\begin{theorem}[Robin scalar boundary estimate]
\label{thm:Robin-Fourier}
Let $q,\widetilde q\in B(Q)$, let
$\widetilde q-q\in L^p(0,1)$, $1<p\le2$, and let $|a|\le A$.
Suppose that the zeros of $F_a(z;q)$ and
$F_a(z;\widetilde q)$ coincide in $|z|<R$, including multiplicities.
Then, with
\[
\beta=\frac{p-1}{3(3p-1)},
\]
one has
\begin{equation}
 |B(0,0)|
 \le C(Q,A,p)
 \bigl(1+\|\widetilde q-q\|_p\bigr)R^{-\beta}.
 \label{eq:corrected-B00}
\end{equation}
In particular, the complete zero data control the scalar boundary
contribution $B(0,0)$ at the rate $R^{-\beta}$.
\end{theorem}

\begin{proof}
The proof has an additional step compared with the Dirichlet argument
in \cite[Section 4]{MSW}.  In the Dirichlet case, equation (11) of
that paper is a direct inverse Fourier transform for
$K_{\widetilde q}(0,t)-K_q(0,t)$.  In the Robin case this is no longer
true because the boundary functional contains a derivative.  The
extra scalar term $B(0,0)\psi(z,0)$ in \eqref{eq:RobinB} must therefore
be isolated before any Fourier inversion is used.

Set
\[
 r=R^\alpha,
 \qquad
 \alpha=\frac{p}{3(3p-1)},
 \qquad
 \beta=\frac{\alpha(p-1)}p
       =\frac{p-1}{3(3p-1)}.
\]
Let $\chi\in C_c^\infty((-2,2))$ be nonnegative, with
$\chi=1$ on $[-1,1]$, and put
\[
 \chi_r(z)=\frac1r\chi\left(\frac zr\right).
\]
Since $\int\chi_r=\int\chi$, equation
\eqref{eq:Robin-Fourier-decomp} gives
\begin{equation}
 B(0,0)
 =\frac1{\int\chi}
   \int_{\mathbb R}\chi_r(z)D_a(z)\,dz
 -\frac1{\int\chi}
   \int_0^2\eta_a(t)\widehat{\chi_r}(-t)\,dt.
 \label{eq:B00-average}
\end{equation}
Because $\chi_r$ is supported in $|z|\le2r$ and $r\ll R^{1/3}$,
the finite-data estimate \eqref{eq:Robin-low-frequency} yields
\begin{align}
 \left|\int\chi_r(z)D_a(z)\,dz\right|
 &\le C R^{-1/3}
   \int_{|z|\le2r}(1+|z|)|\chi_r(z)|\,dz \\
 &\le C R^{-1/3}(1+r)
 \le C R^{-1/3+\alpha}.
 \label{eq:B00-low}
\end{align}
Since $-1/3+\alpha<-\beta$, this contribution is
$O(R^{-\beta})$.

For the second term in \eqref{eq:B00-average}, Hausdorff--Young and
scaling give, with $p'=p/(p-1)$,
\[
 \|\widehat{\chi_r}\|_{L^{p'}(\mathbb R)}
 \le C r^{-1/p'}.
\]
Consequently,
\begin{align}
 \left|\int_0^2\eta_a(t)\widehat{\chi_r}(-t)\,dt\right|
 &\le C\|\eta_a\|_p r^{-1/p'} \\
 &\le C(Q,A,p)
    \bigl(1+\|\widetilde q-q\|_p\bigr)r^{-1/p'}.
 \label{eq:B00-HY}
\end{align}
Since $r^{-1/p'}=R^{-\beta}$, \eqref{eq:corrected-B00} follows.
\end{proof}

\section{Recovery of the Robin boundary trace from the boundary combination}

The Fourier calculation in the preceding section produces the kernel-trace
combination
\[
 \eta_a(t):=K_{q,x}(0,t)-K_{\widetilde q,x}(0,t)
 -a\bigl(K_q(0,t)-K_{\widetilde q}(0,t)\bigr),
\]
not directly the Robin boundary combination
$C_a(t)=B_x(0,t)-aB(0,t)$.  We first pass from $\eta_a$ to $C_a$ by the
Volterra resolvent associated with the transformation operator, and then
recover $B(0,t)$ from the hyperbolic boundary-value problem.  This is the
only modification of the propagation argument in \cite[Section 5]{MSW} that
is forced by the Robin boundary condition.

As in \cite[(21)--(23)]{MSW}, the transformation kernel satisfies
\begin{equation}
 B(x,2-x)=0,\qquad 0\le x\le1,
 \label{eq:Robin-characteristic-boundary}
\end{equation}
and the characteristic representation
\begin{equation}
 B(x,t)=B(0,x+t)+
 \int_{(x+t)/2}^{1}d\alpha
 \int_{(t-x)/2}^{(t+x)/2}
 \bigl(q(\alpha+\beta)-\widetilde q(\alpha-\beta)\bigr)
 B(\alpha-\beta,\alpha+\beta)\,d\beta .
 \label{eq:Robin-Volterra}
\end{equation}
The equation itself is independent of the boundary condition; only the
boundary quantity that is recovered from the spectral data changes.

Put
\[
 b(t):=B(0,t),\qquad c_a(t):=B_x(0,t)-ab(t).
\]
Differentiating \eqref{eq:Robin-Volterra} with respect to $x$ at $x=0$ is
legitimate in the standard weak/Volterra sense used for transformation
kernels.  The inner integration interval collapses at $x=0$, and the two
endpoint derivatives contribute equally.  Hence
\begin{equation}
 B_x(0,t)=b'(t)+\frac12\int_0^{1-t/2}
 \bigl(q(s+t)-\widetilde q(s)\bigr)B(s,s+t)\,ds .
 \label{eq:Robin-boundary-derivative}
\end{equation}
Therefore
\begin{equation}
 b'(t)-ab(t)=c_a(t)-
 \int_0^{1-t/2}
 \bigl(q(s+t)-\widetilde q(s)\bigr)B(s,s+t)\,ds .
 \label{eq:Robin-trace-ODE}
\end{equation}
Moreover, \eqref{eq:Robin-characteristic-boundary} at $x=0$ gives
\begin{equation}
 b(2)=0.
 \label{eq:Robin-b2}
\end{equation}

The transformation identity also gives a Volterra relation between the two
boundary combinations.  Expanding $\psi(z,0)$ and $\psi(z,t)$ in the
transformation representation in \eqref{eq:RobinB} and comparing Fourier
coefficients gives
\begin{equation}
 \eta_a(t)=B(0,0)K_q(0,t)-c_a(t)-
 \int_0^t c_a(s)K_q(s,t)\,ds.
 \label{eq:Robin-eta-c-volterra}
\end{equation}
The corresponding Volterra operator has a resolvent uniformly bounded
in the primitive norm.  We give the needed estimate explicitly.  For
\[
 (Vf)(t)=\int_0^tK_q(s,t)f(s)\,ds,
 \qquad
 P(f)=\sup_{0\le t\le2}\left|\int_t^2f(s)\,ds\right|,
\]
set
\[
 G_t(s)=\int_{\max\{s,t\}}^2K_q(s,u)\,du.
\]
Then Fubini gives
\[
 \int_t^2(Vf)(u)\,du=\int_0^2f(s)G_t(s)\,ds.
\]
The characteristic-coordinate representation above, together with the
explicit formula for $K_{0,q}$, implies the uniform bounded-variation estimate
\begin{equation}
 \sup_{0\le t\le2}\operatorname{Var}_{[0,2]}G_t\le C(Q).
 \label{eq:G-BV}
\end{equation}
Indeed, the $s$-derivative of $K_{0,q}$ contributes
$-\frac14q((s+u)/2)$, whose integral over the triangle is bounded by $Q$;
the derivative of $H_q$ is given by the differentiated version of
\eqref{eq:H-characteristic-representation}, and its total variation is
bounded by $C(Q)$ after one application of Fubini and
\eqref{eq:K-uniform-bound}.  The moving endpoint $u=s$ contributes only
$\int_s^2K_q(s,s)\,ds$, which is also bounded by $C(Q)$.

For smooth $f$, put $F(s)=\int_s^2f(u)\,du$.  Since $f=-F'$, integration
by parts in the Stieltjes sense gives
\[
 \int_0^2f(s)G_t(s)\,ds
 =F(0)G_t(0)+\int_{[0,2]}F(s)\,dG_t(s),
\]
with the endpoint term at $2$ vanishing.  Hence \eqref{eq:G-BV} gives
\begin{equation}
 P(Vf)\le C(Q)P(f).
 \label{eq:P-V}
\end{equation}
Approximation by smooth functions proves the same estimate for $f\in L^1$.
For the iterates, the same calculation is performed on the ordered simplex
$0<s_n<\cdots<s_1<t$; the $n$-fold kernel has total variation bounded by
$C(Q)^n/n!$.  Therefore
\begin{equation}
 P(V^nf)\le\frac{C(Q)^n}{n!}P(f),
 \qquad n\ge1,
 \label{eq:P-Vn}
\end{equation}
and the Neumann series converges in the primitive norm.  In particular,
\begin{equation}
 P((I+V)^{-1}f)\le e^{C(Q)}P(f).
 \label{eq:Robin-primitive-resolvent}
\end{equation}
Applying this to \eqref{eq:Robin-eta-c-volterra}, and using
$P(B(0,0)K_q(0,\cdot))\le C(Q)|B(0,0)|$, gives
\begin{equation}
 \sup_{0\le t\le2}
 \left|\int_t^2 c_a(s)\,ds\right|
 \le C(Q,A)\left(
 |B(0,0)|+
 \sup_{0\le t\le2}\left|\int_t^2\eta_a(s)\,ds\right|\right).
 \label{eq:Robin-c-from-eta}
\end{equation}

We next record the Volterra estimate needed to close
\eqref{eq:Robin-trace-ODE}.  Define
\[
 M(t):=\sup_{t\le\tau\le2}|b(\tau)|.
\]
The same iteration as in Lemma 5.1 of \cite{MSW}, applied to
\eqref{eq:Robin-Volterra}, gives
\begin{equation}
 |B(s,s+t)|\le e^{2Q}M(t),
 \qquad 0\le s\le1-t/2,
 \label{eq:Robin-characteristic-bound}
\end{equation}
for $0\le t\le2$.  Indeed, every boundary value occurring in the iteration
has argument at least $t$, and
$\|q\|_1+\|\widetilde q\|_1\le2Q$.
Consequently,
\begin{equation}
 |b'(t)-ab(t)|
 \le |c_a(t)|+2Qe^{2Q}M(t).
 \label{eq:Robin-trace-differential-estimate}
\end{equation}

The preceding differential estimate must be used at the level of the
primitive of $c_a$, not by taking absolute values of $c_a$.  Put
\[
 H_c(t):=\int_t^2c_a(s)\,ds.
\]
Since $b(2)=0$, variation of constants in \eqref{eq:Robin-trace-ODE},
followed by integration by parts in the $c_a$ term, gives
\begin{equation}
 b(t)=-H_c(t)+a e^{at}\int_t^2e^{-as}H_c(s)\,ds
 +e^{at}\int_t^2e^{-as}J(s)\,ds,
 \label{eq:Robin-trace-primitive-variation}
\end{equation}
where
\[
 J(s):=\int_0^{1-s/2}
 \bigl(q(u+s)-\widetilde q(u)\bigr)B(u,u+s)\,du.
\]
By \eqref{eq:Robin-characteristic-bound},
$|J(s)|\le2Qe^{2Q}M(s)$.  Hence, for $|a|\le A$,
\begin{equation}
 M(t)\le C(A)\sup_{t\le s\le2}|H_c(s)|
       +C(Q,A)\int_t^2M(s)\,ds.
 \label{eq:Robin-trace-primitive-Gronwall}
\end{equation}
Backward Gronwall therefore gives
\begin{equation}
 M(t)\le C(Q,A)\sup_{0\le s\le2}|H_c(s)|.
 \label{eq:Robin-trace-from-h}
\end{equation}
Thus the Robin boundary problem supplies the required map from the primitive
of $c_a$ to the boundary trace, without requiring any pointwise estimate for
$c_a$.

We emphasize the key point: the Fourier argument does not need a
pointwise estimate for $c_a$.  The Robin boundary equation contains $c_a$
only through its primitive.  That primitive gains one power of frequency,
and this is precisely what allows the finite-data estimate for $D_a$ and the
scalar estimate for $B(0,0)$ to close the argument.

\begin{proposition}[Direct estimate for the Robin boundary trace]
\label{prop:Robin-B0-direct}
Let $1<p\le2$, put
\[
 \alpha=\frac{p}{3(3p-1)},
 \qquad
 \beta=\frac{p-1}{3(3p-1)},
 \qquad r=R^\alpha,
\]
and assume the hypotheses of Theorem~\ref{thm:Robin-Fourier}.  Then, for
$R$ sufficiently large and $|a|\le A$,
\begin{equation}
 \sup_{0\le t\le2}|B(0,t)|
 \le C(Q,A,p)\bigl(1+\|\widetilde q-q\|_p\bigr)
 R^{-\beta}.
 \label{eq:Robin-B0-direct}
\end{equation}
In particular, no additional pointwise estimate for $c_a$ is required.
\end{proposition}

\begin{proof}
Put
\[
 H_a(t):=\int_t^2 \eta_a(s)\,ds.
\]
Since
\[
 D_a(z)-B(0,0)=\int_0^2\eta_a(s)e^{izs}\,ds,
\]
and $\eta_a\in L^p(0,2)$ with $p>1$, its Fourier transform belongs to
$L^{p'}(\mathbb R)$ by Hausdorff--Young.  The following inversion identity
is therefore understood first in the distributional sense and then pointwise
for the primitive $H_a$, because multiplication by $1/(iz)$ away from the
origin makes the high-frequency part absolutely integrable in $L^p$--$L^{p'}$
duality.  Fourier inversion gives, for $0\le t\le2$,
\begin{equation}
 H_a(t)=\frac1{2\pi}\int_{\mathbb R}
 \bigl(D_a(z)-B(0,0)\bigr)
 \frac{e^{-izt}-e^{-2iz}}{iz}\,dz,
 \label{eq:Robin-H-Fourier}
\end{equation}
with the quotient at $z=0$ interpreted by continuity.
We split the integral at $|z|=r=R^\alpha$.

For the low-frequency part we exploit the cancellation in the kernel rather
than estimating it by its absolute value.  First,
\begin{equation}
 \int_{|z|\le r}\frac{e^{-izt}-e^{-2iz}}{iz}\,dz
 =2\int_t^2\frac{\sin(rs)}{s}\,ds,
 \label{eq:Robin-H-kernel}
\end{equation}
so its absolute value is bounded by an absolute constant, uniformly in
$0\le t\le2$ and $r\ge1$.  Moreover,
\[
 \int_{|z|\le r}(1+|z|)
 \left|\frac{e^{-izt}-e^{-2iz}}{z}\right|dz
 \le C(1+r).
\]
Consequently, by \eqref{eq:Robin-low-frequency} and
\eqref{eq:corrected-B00},
\begin{align}
 |H_a^{\rm low}(t)|
 &\le C R^{-1/3}(1+r)
 +C|B(0,0)|
 \\&\le C(Q,A,p)\bigl(1+\|\widetilde q-q\|_p\bigr)
 \left(R^{-1/3}r+R^{-\beta}\right).
 \label{eq:Robin-H-low}
\end{align}
Since $r=R^\alpha$ and
$\alpha-1/3< -\beta$, this gives
\begin{equation}
 |H_a^{\rm low}(t)|
 \le C(Q,A,p)\bigl(1+\|\widetilde q-q\|_p\bigr)
 R^{-\beta}.
 \label{eq:Robin-H-low-final}
\end{equation}

For the high-frequency part, use
\eqref{eq:Robin-Fourier-decomp} and Hausdorff--Young.  Since
\[
 \left|\frac{e^{-izt}-e^{-2iz}}{z}\right|\le\frac{2}{|z|},
\]
we obtain, with $p'=p/(p-1)$,
\begin{align}
 |H_a^{\rm high}(t)|
 &\le C\|\eta_a\|_p
 \left(\int_{|z|>r}|z|^{-p}\,dz\right)^{1/p}
 \\&\le C(Q,A,p)
 \bigl(1+\|\widetilde q-q\|_p\bigr)r^{-1/p'}
 \\&=C(Q,A,p)
 \bigl(1+\|\widetilde q-q\|_p\bigr)R^{-\beta}.
 \label{eq:Robin-H-high}
\end{align}
Combining \eqref{eq:Robin-H-low-final} and \eqref{eq:Robin-H-high}
proves
\begin{equation}
 \sup_{0\le t\le2}|H_a(t)|
 \le C(Q,A,p)\bigl(1+\|\widetilde q-q\|_p\bigr)
 R^{-\beta}.
 \label{eq:Robin-H-final}
\end{equation}

It remains to pass from $H_a$ to $b(t)=B(0,t)$.  By
\eqref{eq:Robin-c-from-eta},
\[
 \sup_{0\le t\le2}\left|\int_t^2c_a(s)\,ds\right|
 \le C(Q,A)\left(|B(0,0)|+\sup_{0\le t\le2}|H_a(t)|\right).
\]
Variation of constants in \eqref{eq:Robin-trace-ODE}, followed by backward
Gronwall, therefore gives
\[
 \sup_{0\le t\le2}|B(0,t)|
 \le C(Q,A)\left(|B(0,0)|+\sup_{0\le t\le2}|H_a(t)|\right).
\]
Together with \eqref{eq:corrected-B00} and \eqref{eq:Robin-H-final}, this proves
\eqref{eq:Robin-B0-direct}.
\end{proof}

\section{Propagation from the boundary trace}

We now propagate the estimate for $B(0,t)$ into the characteristic triangle.
This is exactly the Volterra argument of \cite[Section 5]{MSW}; we include the
argument so that the dependence on the new Robin exponent is explicit.

\begin{lemma}[Uniform Volterra propagation]
\label{lem:uniform-volterra}
Suppose that, for some $C_*$,
\begin{equation}
 |B(0,t)|\le C_*,
 \qquad 0\le t\le2.
 \label{eq:uniform-volterra-assumption}
\end{equation}
Then, in the triangle $0\le x\le t$, $x+t\le2$,
\begin{equation}
 |B(x,t)|\le C(Q)C_*.
 \label{eq:uniform-B-interior}
\end{equation}
\end{lemma}

\begin{proof}
Write
\[
 B=\sum_{n=0}^{\infty}B_n,
 \qquad B_0(x,t)=B(0,x+t),
\]
with the iterates generated by \eqref{eq:Robin-Volterra}.  The iteration from
Lemma 5.1 of \cite{MSW} gives, for $n\ge1$,
\begin{equation}
 |B_n(x,t)|
 \le
 C_*\frac{(2Q)^n}{(n-1)!}
 \left(1-\frac{x+t}{2}\right)^{n-1}.
 \label{eq:uniform-Bn}
\end{equation}
Summing the series and using $|B_0(x,t)|\le C_*$ yields
\eqref{eq:uniform-B-interior}, with $C(Q)$ depending only on $Q$.
\end{proof}

We can now propagate the uniform boundary estimate.  By
Proposition~\ref{prop:Robin-B0-direct},
\begin{equation}
 |B(0,t)|\le C(Q,A,p)\bigl(1+\|\widetilde q-q\|_p\bigr)R^{-\beta},
 \qquad 0\le t\le2.
 \label{eq:Robin-B0-uniform-final}
\end{equation}
Applying Lemma~\ref{lem:uniform-volterra} therefore gives
\begin{equation}
 |B(x,t)|
 \le C(Q,A,p)\bigl(1+\|\widetilde q-q\|_p\bigr)
 R^{-\beta},
 \qquad 0\le x\le t,\quad x+t\le2.
 \label{eq:corrected-B-final}
\end{equation}

\begin{theorem}[Finite-data stability for the potential]
\label{thm:mainstability}
Let $1<p\le2$, let $q,\widetilde q\in B(Q)$, and suppose
$\|\widetilde q-q\|_p\le Q_p$.  Assume the zero sets of
$F_a(z;q)$ and $F_a(z;\widetilde q)$ coincide in $|z|<R$, including
multiplicities, and let $|a|\le A$.  Then, for sufficiently large $R$,
\begin{equation}
 \boxed{
 \sup_{0\le x\le1}
 \left|\int_x^1(\widetilde q-q)(s)\,ds\right|
 \le C(Q,A,Q_p,p)
 R^{-\beta},
 \qquad
 \beta=\frac{p-1}{3(3p-1)}.
 }
 \label{eq:mainstability}
\end{equation}
\end{theorem}

\begin{proof}
The diagonal identity gives
\[
 \left|\int_x^1(\widetilde q-q)(s)\,ds\right|=2|B(x,x)|.
\]
For $x\in[0,1]$, the point $(x,x)$ belongs to the characteristic
triangle, so \eqref{eq:corrected-B-final} yields the claimed estimate
immediately.
\end{proof}

\section{Perturbed Robin eigenvalues and resonances}

We now allow the Robin eigenvalues and resonances to be perturbed.  The
argument is the same canonical-product argument as in \cite[Section~6]{MSW}.
The only change is that the Robin characteristic function has linear rather
than bounded high-energy growth.  Accordingly, in the final Fourier
primitive we retain the factor $(1+|z|)$ from \eqref{eq:realbound}.
No new transformation-operator argument is needed in this section.

Let $q,\widetilde q\in B(Q)$, $|a|\le A$, and let
$\{z_n\}$ and $\{\widetilde z_n\}$ denote the zeros of
$F_a(z;q)$ and $F_a(z;\widetilde q)$, counted with multiplicity.  Assume
that the nonzero zeros in $|z|<R$ admit a multiplicity-preserving pairing such
that
\begin{equation}
 |z_n-\widetilde z_n|\le\varepsilon,
 \qquad 0<\varepsilon<\frac34,
 \label{eq:Robin-zero-perturbation}
\end{equation}
and assume that the multiplicities of the zeros at the origin agree.  We
write $N_a(r)$ for the number of zeros in $|z|\le r$.

The zero-counting estimate \eqref{eq:zero-count} gives
\begin{equation}
 N_a(r)\le C(Q,A)(1+r).
 \label{eq:Robin-perturbed-count}
\end{equation}
The high-energy asymptotics, uniformly for $\|q\|_1,\|\widetilde q\|_1\le Q$
and $|a|\le A$, imply that there is a single constant $C_0=C(Q,A)$ such that
all zeros of both $F_a(\cdot;q)$ and $F_a(\cdot;\widetilde q)$ in the upper
half-plane lie in $|z|\le C_0$.  We use the contour $\mathcal C_r$ consisting
of the horizontal segment at height $1$ (with the usual joining pieces for
the Cauchy deformation) and the fixed upper semicircular detour of radius
$C_0+1$ in the middle.  Since the same $C_0$ applies to both potentials, every
upper-half-plane zero $z_n$ of $F_a(\cdot;q)$ and every upper-half-plane zero
$\widetilde z_n$ of $F_a(\cdot;\widetilde q)$ satisfies
\[
 \operatorname{dist}(z_n,\text{semicircle})\ge1,
 \qquad
 \operatorname{dist}(\widetilde z_n,\text{semicircle})\ge1.
\]
This separation is a direct consequence of $|z_n|,|\widetilde z_n|\le C_0$
and does not use the pairing error $\varepsilon$.  Zeros in the lower
half-plane are separated from the line $\operatorname{Im}z=1$ by at least
one, so no cluster-by-cluster detours are needed.

For
\begin{equation}
 r=R^\alpha,
 \qquad
 \alpha=\frac{p}{3(3p-1)},
 \qquad
 \beta=\frac{p-1}{3(3p-1)},
 \label{eq:Robin-perturbed-scales}
\end{equation}
we have $r=o(R^{1/3})$ and
\begin{equation}
 -\frac13+2\alpha=-\beta,
 \qquad
 r^{-1/p'}=R^{-\beta},
 \quad p'=\frac p{p-1}.
 \label{eq:Robin-perturbed-scale-identities}
\end{equation}

\begin{lemma}[Perturbed Robin canonical-product estimate]
\label{lem:Robin-perturbed-characteristic}
Let $Q,A>0$.  There are constants $C=C(Q,A)$ and $R_0=R_0(Q,A)$ such
that, under \eqref{eq:Robin-zero-perturbation}, the quotient of the finite
paired genus-one products satisfies, on the contour $\mathcal C_r$ described
above,
\begin{equation}
 \int_{\mathcal C_r}
 \left|\log W_R^0(z)\right|\,|dz|
 \le C\varepsilon r\log R,
 \label{eq:Robin-logW-integral}
\end{equation}
and
\begin{equation}
 \sup_{z\in\mathcal C_r}|W_R^0(z)|
 \le \exp(C\varepsilon r).
 \label{eq:Robin-W-sup}
\end{equation}
Here $W_R^0$ is the quotient of the paired genus-one factors for
$0<|z_n|<R$; the common factor $z^m$ at the origin is removed first.
Slightly abusing notation, throughout this factorization the cutoff at $R$
is understood pairwise: both members of any paired pair are assigned to the
same one of the finite product and the tail product.  Thus no pair is split
by the radial cutoff.  The affine Hadamard factor is treated separately below.
Moreover, if
\begin{equation}
 \rho=R^{1/3},
 \label{eq:Robin-rho}
\end{equation}
then, on the disk $|z-3i\rho|\le\rho$,
\begin{equation}
 |W_R^0(z)-1|\le C\varepsilon\log R\,e^{C\varepsilon\log R}.
 \label{eq:Robin-W-high}
\end{equation}
Finally, after increasing $R_0$ if necessary, the affine quotient
\begin{equation}
 E_R(z):=e^{g(z)-\widetilde g(z)}
\end{equation}
satisfies, on $|z|\le r$,
\begin{equation}
 |E_R(z)-1|
 \le C\left(\varepsilon\log R+R^{-1/3}\right)
 \exp\!\left(C\varepsilon\log R+C R^{-1/3}\right).
 \label{eq:Robin-affine-low}
\end{equation}
The same bound, with the right-hand side interpreted exponentially, is valid
without the smallness assumption $\varepsilon\log R\le c_0$; when
$\varepsilon\log R>c_0$ the final stability estimate follows from the
trivial $L^p$ bound after enlarging the constant.
\end{lemma}

\begin{proof}
After the common factor at zero is removed, Hadamard factorization gives
\begin{equation}
 F_a(z;q)=e^{g(z)}P_{<R}(z)P_{\ge R}(z),
 \qquad
 F_a(z;\widetilde q)=e^{\widetilde g(z)}
 \widetilde P_{<R}(z)\widetilde P_{\ge R}(z),
 \label{eq:Robin-Hadamard-split}
\end{equation}
where $g-\widetilde g$ is affine.  For a paired zero the genus-one ratio is
\begin{equation}
 \frac{E(z/z_n)}{E(z/\widetilde z_n)}
 =\frac{z-z_n}{z-\widetilde z_n}
   \frac{\widetilde z_n}{z_n}
   \exp\!\left[z\left(\frac1{z_n}-\frac1{\widetilde z_n}\right)\right].
 \label{eq:Robin-E-ratio}
\end{equation}
For $|z_n|\ge1$, the pairing assumption gives
$|\widetilde z_n|\ge |z_n|/4$ and hence
\begin{equation}
 \left|\frac1{z_n}-\frac1{\widetilde z_n}\right|
 \le C\frac{\varepsilon}{|z_n|^2}.
 \label{eq:Robin-reciprocal-difference}
\end{equation}
The finitely many zeros with $|z_n|<1$ are treated as complete factors; their
number is $O_{Q,A}(1)$ and they cause no singular reciprocal estimate.

On the horizontal part at height $1$, lower-half-plane zeros are at distance
at least one, and the fixed semicircle is at distance at least one from every
upper-half-plane zero of either characteristic function, by the preceding
uniform disk bound.  Thus
\[
 \left|\log\frac{z-z_n}{z-\widetilde z_n}\right|
 \le C\frac{\varepsilon}{|z-\widetilde z_n|}.
\]
Split the zeros into $|z_n|\le2r$ and the dyadic shells
$2^jr<|z_n|\le2^{j+1}r$.  The first group contains $O_{Q,A}(r)$ zeros.
For each such zero, the separation of the contour from the zeros and the
standard estimate for the integral of $|z-\widetilde z_n|^{-1}$ along a
piecewise smooth contour of length $O(r)$ give an $O(\log R)$ bound.
Thus the first group contributes $O(\varepsilon r\log R)$.  In the $j$th
shell, each factor contributes $O(\varepsilon/(2^jr))$ pointwise and the
shell contains $O_{Q,A}(2^jr)$ zeros.  Since the contour has length $O(r)$,
each nonempty shell contributes $O(\varepsilon r)$.  There are
$O(\log R)$ shells, giving
\eqref{eq:Robin-logW-integral}.  The exponential corrections in
\eqref{eq:Robin-E-ratio} are bounded by $C\varepsilon|z|$ because
\[
 \sum_{1\le|z_n|<R}
 \left|\frac1{z_n}-\frac1{\widetilde z_n}\right|
 \le C\varepsilon,
\]
so they contribute $O(\varepsilon r)$ on $\mathcal C_r$.
This proves \eqref{eq:Robin-W-sup}.

For the high-imaginary disk $D_\rho:=\{|z-3i\rho|\le\rho\}$,
the Robin high-energy estimate gives
\begin{equation}
 \frac{F_a(z;q)}{F_a(z;\widetilde q)}=1+O(\rho^{-1}),
 \qquad z\in D_\rho,
 \label{eq:Robin-ratio-high}
\end{equation}
uniformly for $|a|\le A$.  The tail products satisfy, by the usual genus-one
estimate and the zero-counting bound,
\begin{equation}
 P_{\ge R}(z)=1+O(\rho^2/R),
 \qquad
 \widetilde P_{\ge R}(z)=1+O(\rho^2/R).
 \label{eq:Robin-tail-high}
\end{equation}
The same dyadic-shell argument as above, now centered at $3i\rho$, gives
\[
 |\log W_R^0(z)|\le C\varepsilon\log R,
 \qquad z\in D_\rho,
\]
and hence \eqref{eq:Robin-W-high}.  Combining
\eqref{eq:Robin-Hadamard-split}--\eqref{eq:Robin-tail-high} gives
\begin{equation}
 E_R(z)=1+O\!\left(\varepsilon\log R+R^{-1/3}\right),
 \qquad z\in D_\rho,
 \label{eq:Robin-affine-high}
\end{equation}
provided $\varepsilon\log R$ is below a fixed small constant.  This is the
only place where the affine Hadamard factor enters.

Write $\delta_R=C(\varepsilon\log R+R^{-1/3})$, with the constant
chosen so that \eqref{eq:Robin-affine-high} gives $|E_R(z)-1|\le\delta_R$ on
$D_\rho$.  In the smallness regime under consideration, take $R$ large
(and the threshold for $\varepsilon\log R$ small) so that $\delta_R<1/2$.
Then $E_R(D_\rho)$ lies in the disk $|w-1|<1/2$, where the principal
logarithm is analytic.  Consequently
\[
 h_R(z):=\operatorname{Log}E_R(z)
 \quad\text{is analytic on }D_\rho,
 \qquad |h_R(z)|\le 2|E_R(z)-1|\le 2\delta_R.
\]
Since $E_R(z)=\exp(c_1z+c_0)$, we have $h_R'(z)=c_1$ on this disk; the
possible additive multiple of $2\pi i$ between $h_R$ and $c_1z+c_0$ does
not affect the derivative or the exponential.  Cauchy's estimate therefore
gives
\[
 |c_1|\le C\rho^{-1}
       \left(\varepsilon\log R+R^{-1/3}\right).
\]
Using the bound on $h_R$ at the disk center and this derivative estimate,
we obtain $|h_R(z)|\le C(\varepsilon\log R+R^{-1/3})$ for $|z|\le
r=o(\rho)$.  Exponentiating yields \eqref{eq:Robin-affine-low}.  If $\varepsilon\log R$ is not small, the
right-hand side of the final stability estimate is already bounded below by
a fixed positive constant (after enlarging constants), so the trivial estimate
$\sup_x|\int_x^1(\widetilde q-q)|\le Q_p$ suffices.
\end{proof}

\begin{lemma}[Perturbed Robin Fourier integral]
\label{lem:Robin-perturbed-D}
Let
\[
 k_t(z):=\frac{e^{-izt}-e^{-2iz}}{iz},
 \qquad 0\le t\le2,
\]
with the value at $z=0$ defined by continuity.  Under the hypotheses of
Lemma~\ref{lem:Robin-perturbed-characteristic},
\begin{equation}
 \sup_{0\le t\le2}
 \left|\int_{-r}^{r}D_a(z)k_t(z)\,dz\right|
 \le C(Q,A)R^{-\beta}
 +C(Q,A)\varepsilon r\log R\,e^{C\varepsilon r\log R}.
 \label{eq:Robin-perturbed-Fourier-integral}
\end{equation}
\end{lemma}

\begin{proof}
Because $D_a$ is entire, deform the integral over $[-r,r]$ to the contour
$\mathcal C_r$.  Write
\begin{equation}
 \frac{F_a(z;q)}{F_a(z;\widetilde q)}
 =E_R(z)W_R^0(z)T_R(z),
 \label{eq:Robin-perturbed-factorization}
\end{equation}
where $T_R=P_{\ge R}/\widetilde P_{\ge R}$.  On $\mathcal C_r$,
\[
 |F_a(z;\widetilde q)|\le C(1+|z|),
 \qquad
 |k_t(z)|\le C\min\!\left(1,\frac1{|z|}\right),
\]
with the first bound using the standard entire-function estimate and the second
following directly from the definition of $k_t$ (the fixed detours near zero
are harmless).  Therefore
\[
 |F_a(z;\widetilde q)k_t(z)|\le C
\]
on the whole contour, after increasing the constant to cover $|z|\le1$.

From \eqref{eq:Robin-logW-integral}, \eqref{eq:Robin-W-sup}, and
$|e^u-1|\le|u|e^{|u|}$,
\[
 \int_{\mathcal C_r}|W_R^0(z)-1|\,|dz|
 \le C\varepsilon r\log R\,e^{C\varepsilon r\log R}.
\]
The affine estimate \eqref{eq:Robin-affine-low} contributes
\[
 C r\left(\varepsilon\log R+R^{-1/3}\right),
\]
which is bounded by
$C\varepsilon r\log R+C R^{-1/3}r$.
For the tail quotient,
\[
 |T_R(z)-1|\le C\frac{|z|^2}{R},
\]

after increasing $R_0$ so that $r^2/R$ is small.  Hence its contribution is
\[
 C\frac{r^3}{R}.
\]
Since $r=R^\alpha$,
\[
 \frac{r^3}{R}=R^{3\alpha-1}
 \le R^{-\beta},
\]
because
\begin{equation}
 3\alpha-1\le-\beta,
 \label{eq:Robin-tail-exponent}
\end{equation}
which is the correct exponent comparison.  Also
$R^{-1/3}r\le R^{-\beta}$ because $\alpha-1/3<-\beta$.
Combining the three contributions proves the lemma.
\end{proof}

\begin{lemma}[Perturbed Robin Fourier primitive]
\label{lem:Robin-perturbed-H}
Let
\[
 H_a(t)=\int_t^2\eta_a(s)\,ds,
 \qquad
 D_a(z)-B(0,0)=\int_0^2\eta_a(s)e^{izs}\,ds.
\]
Under the hypotheses of Lemma~\ref{lem:Robin-perturbed-characteristic}, if
$\|\widetilde q-q\|_p\le Q_p$, then
\begin{equation}
 \sup_{0\le t\le2}|H_a(t)|
 \le C(Q,A,p,Q_p)
 \left[R^{-\beta}
 +\varepsilon r\log R\,e^{C\varepsilon r\log R}\right].
 \label{eq:Robin-perturbed-H-bound}
\end{equation}
\end{lemma}

\begin{proof}
Fourier inversion for the primitive gives
\begin{equation}
 H_a(t)=\frac1{2\pi}\int_{\mathbb R}
 \bigl(D_a(z)-B(0,0)\bigr)
 \frac{e^{-izt}-e^{-2iz}}{iz}\,dz.
 \label{eq:Robin-perturbed-H-Fourier}
\end{equation}
The quotient at $z=0$ is interpreted by continuity.  Split the integral at
$|z|=r$.

For $|z|>r$, Hausdorff--Young and
$|k_t(z)|\le2/|z|$ give
\[
 |H_a^{\rm high}(t)|
 \le C\|\eta_a\|_p r^{-1/p'}
 \le C(Q,A,p,Q_p)R^{-\beta}.
\]

For the low-frequency part, apply
Lemma~\ref{lem:Robin-perturbed-D} to the $D_a$ term.  The scalar term is estimated independently by averaging the same identity
uniformly over $[-r,r]$:
\begin{equation}
 B(0,0)=\frac1{2r}\int_{-r}^{r}D_a(z)\,dz
 -\frac1{2r}\int_0^2\eta_a(t)
 \left(\int_{-r}^{r}e^{izt}\,dz\right)dt.
 \label{eq:Robin-perturbed-B00-average}
\end{equation}
The first integral is entire and can therefore be deformed to $\mathcal C_r$.
Using \eqref{eq:Robin-perturbed-factorization} and the same estimates as in
the proof of Lemma~\ref{lem:Robin-perturbed-D}, with the extra factor
$1/r$, gives
\[
 \frac1{2r}\left|\int_{-r}^{r}D_a(z)\,dz\right|
 \le C\left[\varepsilon r\log R\,e^{C\varepsilon r\log R}
 +R^{-1/3}r+\frac{r^3}{R}\right].
\]
For the second term, Hausdorff--Young (or directly the scaling of the sinc
kernel) gives
\[
 \frac1{2r}\left|\int_0^2\eta_a(t)
 \left(\int_{-r}^{r}e^{izt}\,dz\right)dt\right|
 \le C\|\eta_a\|_p r^{-1/p'}
 \le C(Q,A,p,Q_p)R^{-\beta}.
\]
Using $R^{-1/3}r\le R^{-\beta}$ and
$r^3/R\le R^{-\beta}$ gives
\begin{equation}
 |B(0,0)|
 \le C(Q,A,p,Q_p)
 \left[R^{-\beta}
 +\varepsilon r\log R\,e^{C\varepsilon r\log R}\right].
 \label{eq:Robin-perturbed-B00}
\end{equation}
Finally, the low-frequency scalar contribution in
\eqref{eq:Robin-perturbed-H-Fourier} is bounded by
$C|B(0,0)|$, because
\[
 \sup_{0\le t\le2}
 \left|\int_{-r}^{r}k_t(z)\,dz\right|\le C.
\]
Combining the low- and high-frequency estimates proves the lemma.
\end{proof}

\begin{theorem}[Perturbed Robin inverse-resonance stability]
\label{thm:Robin-perturbed-stability}
Let $Q_1,Q_p,A>0$ and $1<p\le2$.  There are constants
$C=C(Q_1,Q_p,A,p)$ and $R_0=R_0(Q_1,Q_p,A,p)$ such that, for
$R\ge R_0$ and $0<\varepsilon<3/4$, the following holds.  Suppose
$q,\widetilde q\in B(Q_1)$,
\[
 \|\widetilde q-q\|_{L^p(0,1)}\le Q_p,
 \qquad |a|\le A,
\]
and the nonzero zeros of $F_a(z;q)$ and $F_a(z;\widetilde q)$ in
$|z|<R$ can be paired, with multiplicities, so that
\eqref{eq:Robin-zero-perturbation} holds.  If $0$ is a zero, its multiplicity
is assumed to be the same for both characteristic functions.  Then
\begin{equation}
 \boxed{
 \sup_{0\le x\le1}
 \left|\int_x^1(\widetilde q-q)(s)\,ds\right|
 \le C
 \left[R^{-\beta}
 +\varepsilon R^\alpha\log R\,
 e^{C\varepsilon R^\alpha\log R}\right],
 }
 \label{eq:Robin-perturbed-main}
\end{equation}
where
\[
 \alpha=\frac{p}{3(3p-1)},
 \qquad
 \beta=\frac{p-1}{3(3p-1)}.
 \]
\end{theorem}

\begin{remark}[Perturbative regime]
The perturbative term is informative when $\varepsilon R^\alpha\log R$ is bounded.
In particular, if $\varepsilon R^\alpha\log R\le c_0$, then
\[
 R^{-\beta}+
 \varepsilon R^\alpha\log R\,e^{C\varepsilon R^\alpha\log R}
 \le C_{c_0}\bigl(R^{-\beta}+\varepsilon R^\alpha\log R\bigr).
\]
\end{remark}

\begin{proof}
Lemma~\ref{lem:Robin-perturbed-H} gives the bound for the primitive of the
Robin boundary combination.  The backward Gronwall estimate
\eqref{eq:Robin-trace-from-h} therefore gives
\[
 \sup_{0\le t\le2}|B(0,t)|
 \le C(Q_1,A,p,Q_p)
 \left[R^{-\beta}+\varepsilon R^\alpha\log R\,e^{C\varepsilon R^\alpha\log R}\right].
\]
The Volterra propagation lemma \ref{lem:uniform-volterra} gives the same
bound for $B(x,t)$ in the characteristic triangle.  Finally,
\eqref{eq:diagonal} yields
\[
 \left|\int_x^1(\widetilde q-q)(s)\,ds\right|=2|B(x,x)|,
\]
which proves \eqref{eq:Robin-perturbed-main}.
\end{proof}

\section{Uniqueness}

The stability theorem gives the quantitative $L^p$ result.  Complete-data
uniqueness itself does not require the difference of the potentials to lie in
$L^p$, because the Volterra representation above gives an $L^1$ Fourier
representation of the characteristic-function difference.

\begin{corollary}[Uniqueness in $L^1$ for fixed Robin parameter]
\label{cor:uniqueness-L1}
Let $q,\widetilde q\in B(Q)$ and let $a\in\mathbb C$ be fixed.  If the
complete sets of zeros of $F_a(z;q)$ and $F_a(z;\widetilde q)$ coincide,
including multiplicities, then
\[
 q=\widetilde q\qquad\text{a.e. on }[0,1].
\]
\end{corollary}

\begin{proof}
Hadamard factorization together with the common high-energy normalization
shows that the two characteristic functions are equal.  On the real axis
we have
\[
0=D_a(z)=B(0,0)+\int_0^2\eta_a(t)e^{izt}\,dt,
\]
where $\eta_a\in L^1(0,2)$ by the direct Jost-Volterra representation
used in Lemma~\ref{lem:Robin-Fourier-decomp}.  The Riemann--Lebesgue lemma
therefore gives $B(0,0)=0$, and uniqueness of the Fourier transform on
$L^1$ gives $\eta_a=0$ a.e.  Equation
\eqref{eq:Robin-eta-c-volterra}, followed by the Volterra resolvent estimate
\eqref{eq:Robin-primitive-resolvent}, gives $c_a=0$.  The boundary equation
\eqref{eq:Robin-trace-ODE} then has only the Volterra term $J(t)$ on its
right-hand side.  The estimate \eqref{eq:Robin-characteristic-bound} yields
$M(t)\le C(Q,A)\int_t^2M(s)\,ds$, so backward Gronwall gives
$M(t)=0$ and hence $b=B(0,\cdot)=0$.  The Volterra propagation
\eqref{eq:Robin-Volterra} then gives $B=0$ throughout the characteristic
triangle.  Finally, \eqref{eq:diagonal} implies
$\int_x^1(\widetilde q-q)=0$ for every $x$, and therefore
$q=\widetilde q$ a.e.
\end{proof}

\begin{corollary}[Quantitative uniqueness consequence]
\label{cor:uniqueness}
Let \(q,\widetilde q\in B(Q)\) with
\(\widetilde q-q\in L^p(0,1)\) for some \(1<p\le2\), and suppose the
\(L^p\)-difference has a fixed bound.  Let \(a\in\mathbb C\) be fixed, and suppose
that the complete sets of zeros of
\[
 F_a(z;q)
 \quad\text{and}\quad
 F_a(z;\widetilde q)
\]
coincide, including multiplicities.  Then
\[
 q=\widetilde q
 \qquad\text{a.e. on }[0,1].
\]
\end{corollary}

\begin{proof}
The zero sets coincide in every disc \(|z|<R\).  For every sufficiently
large \(R\), Theorem~\ref{thm:mainstability} therefore gives
\[
 \sup_{0\le x\le1}
 \left|\int_x^1(\widetilde q-q)(t)\,dt\right|
 \le C(Q,A,Q_p,p)R^{-\beta},
\]
whenever \(\widetilde q-q\in L^p(0,1)\) with a fixed bound \(Q_p\).
Letting \(R\to\infty\) yields
\[
 \int_x^1(\widetilde q-q)(t)\,dt=0,
 \qquad 0\le x\le1.
\]
Hence \(\widetilde q=q\) almost everywhere.
\end{proof}

\section{Remarks on the complex Robin parameter}

The parameter \(a\) is allowed to be complex throughout.  No conjugation of
\(a\) occurs in the definition of the characteristic function or in the
transformation identity.  The constants in the fixed-parameter stability
estimates are uniform for
\[
 |a|\le A.
\]

If \(a\in\mathbb R\) and \(q\) is real-valued, the corresponding realization
is self-adjoint and the usual spectral restrictions apply.  For general
\(q\in B(Q)\) and \(a\in\mathbb C\), no self-adjoint symmetry of the zero set
is assumed or required.

We emphasize the distinction between the stability problem for a fixed Robin
parameter and the inverse problem in which the Robin parameter is itself
unknown.  For fixed \(a\), the characteristic function has the high-energy
normalization
\[
 F_a(z;q)=iz-a+O(1),
\]
with the more precise constant-term expansion
\[
 F_a(z;q)=iz-a-\frac12\int_0^1q(t)\,dt+o(1).
\]
Thus the role of the Dirichlet normalization is taken by the free Robin
expression \(iz-a\), while the next asymptotic coefficient contains both \(a\)
and information about the potential.

Suppose now that two characteristic functions, with parameters \(a\) and
\(\widetilde a\), have the same complete zero set, including multiplicities.
Hadamard factorization then gives
\[
 \frac{F_a(z;q)}{F_{\widetilde a}(z;\widetilde q)}
 =C e^{cz}.
\]
The common leading behavior \(F_a(z;q)=iz+O(1)\) along the positive
imaginary axis forces \(c=0\) and \(C=1\).  Hence equality of the complete
zero sets implies
\[
 F_a(z;q)=F_{\widetilde a}(z;\widetilde q).
\]
Comparing the constant terms gives the necessary scalar identity
\begin{equation}
 \boxed{
 \widetilde a-a
 +\frac12\int_0^1(\widetilde q(t)-q(t))\,dt=0.
 }
 \label{eq:joint-identifiability}
\end{equation}

Equation \eqref{eq:joint-identifiability} is only one scalar relation and is
not by itself a uniqueness theorem for the pair \((a,q)\).  If the Robin
parameter is unknown, the full identity
\[
 F_a(z;q)=F_{\widetilde a}(z;\widetilde q)
\]
must be combined with an additional identifiability argument.  The present
finite-data stability and uniqueness theorem is deliberately formulated for
fixed \(a\), where the complete resonance data determine \(q\) uniquely.

When the same parameter \(a\) is used for both operators, the free factor
\(iz-a\) is common to the two characteristic functions.  Thus no additional
boundary-parameter term occurs in the finite-data argument of Sections~4--9.

\section{Concluding remarks}

The argument developed here shows that the inverse resonance method of
Marletta--Shterenberg--Weikard \cite{MSW} extends to a fixed complex Robin
boundary condition.  The analytic characteristic function is normalized at
infinity by
\[
 iz-a,
\]
rather than by the Dirichlet normalization.  This changes the
order-one bookkeeping in the canonical product but leaves the transformation
operator and Volterra structure intact.

The essential new boundary quantity is
\[
 h_a(t)=B_x(0,t)-aB(0,t).
\]
The Fourier calculation controls its primitive rather than directly
recovering \(B(0,t)\).  The characteristic boundary equation and the backward
Gronwall argument then recover the boundary trace, while the standard
Volterra iteration propagates that estimate into the interior.  Crucially, the
cancellation in the Fourier primitive estimate means that no artificial
factor \(\log R\) is introduced in this boundary reconstruction.

For coincident finite resonance data this yields the quantitative estimate
\[
 \sup_{0\le x\le1}
 \left|\int_x^1(\widetilde q-q)(t)\,dt\right|
 \le C(Q,A,Q_p,p)R^{-\beta},
 \qquad
 \beta=\frac{p-1}{3(3p-1)},
\]
which is the main stability result of the paper.  Letting \(R\to\infty\)
gives uniqueness for the potential when the Robin parameter is fixed.

For perturbed resonance data, the boundary and interior propagation steps are
already available in the form of Lemma~\ref{lem:Robin-perturbed-H} together with the boundary reconstruction estimate \eqref{eq:Robin-trace-from-h}.
The perturbed-data theorem is obtained in Section~10 by estimating that
primitive directly from the perturbed canonical product, using the cancellation
provided by the factor $1/z$ in its Fourier representation.

Finally, when \(a\) is unknown, the zero set determines the characteristic
function after the Jost normalization is fixed, but its constant high-energy
coefficient yields only the scalar relation \eqref{eq:joint-identifiability}.
Additional work is required to establish joint identifiability of \(a\) and
\(q\).  The results proved here therefore concern the fixed complex Robin
parameter problem.

\end{document}